\documentclass[twocolumn,secnumarabic,aps,superscriptaddress]{revtex4-1}
\usepackage{bbm,amsmath,url,graphicx,amssymb,mathtools}
\usepackage[colorlinks,pdftex]{hyperref} 
\usepackage{mathrsfs}
\usepackage{amsthm}

\usepackage[dvipsnames]{xcolor}
\newtheorem{thm}{Theorem}

\newtheorem{cor}[thm]{Corollary}
\newtheorem{lem}[thm]{Lemma}
\newtheorem{prop}[thm]{Proposition}

\newtheorem*{remark}{Remark}

\usepackage{prettyref}
\newcommand{\savehyperref}[2]{\texorpdfstring{\hyperref[#1]{#2}}{#2}}

\newrefformat{eq}{\savehyperref{#1}{\textup{(\ref*{#1})}}}
\newrefformat{lem}{\savehyperref{#1}{Lemma~\ref*{#1}}}
\newrefformat{def}{\savehyperref{#1}{Definition~\ref*{#1}}}
\newrefformat{thm}{\savehyperref{#1}{Theorem~\ref*{#1}}}
\newrefformat{cor}{\savehyperref{#1}{Corollary~\ref*{#1}}}
\newrefformat{cha}{\savehyperref{#1}{Chapter~\ref*{#1}}}
\newrefformat{sec}{\savehyperref{#1}{Section~\ref*{#1}}}
\newrefformat{app}{\savehyperref{#1}{Appendix~\ref*{#1}}}
\newrefformat{tab}{\savehyperref{#1}{Table~\ref*{#1}}}
\newrefformat{fig}{\savehyperref{#1}{Figure~\ref*{#1}}}
\newrefformat{hyp}{\savehyperref{#1}{Hypothesis~\ref*{#1}}}
\newrefformat{alg}{\savehyperref{#1}{Algorithm~\ref*{#1}}}
\newrefformat{rem}{\savehyperref{#1}{Remark~\ref*{#1}}}
\newrefformat{item}{\savehyperref{#1}{Item~\ref*{#1}}}
\newrefformat{step}{\savehyperref{#1}{step~\ref*{#1}}}
\newrefformat{conj}{\savehyperref{#1}{Conjecture~\ref*{#1}}}
\newrefformat{fact}{\savehyperref{#1}{Fact~\ref*{#1}}}
\newrefformat{prop}{\savehyperref{#1}{Proposition~\ref*{#1}}}
\newrefformat{prob}{\savehyperref{#1}{Problem~\ref*{#1}}}
\newrefformat{claim}{\savehyperref{#1}{Claim~\ref*{#1}}}
\newrefformat{relax}{\savehyperref{#1}{Relaxation~\ref*{#1}}}
\newrefformat{red}{\savehyperref{#1}{Reduction~\ref*{#1}}}
\newrefformat{part}{\savehyperref{#1}{Part~\ref*{#1}}}

\newcommand{\bra}[1]{\langle #1|}
\newcommand{\ket}[1]{|#1\rangle}

\DeclareMathOperator{\tr}{tr}
\def\eps{\epsilon}

\mathchardef\ordinarycolon\mathcode`\:
\mathcode`\:=\string"8000
\def\vcentcolon{\mathrel{\mathop\ordinarycolon}}
\begingroup \catcode`\:=\active
  \lowercase{\endgroup
  \let :\vcentcolon
  }

\newcommand{\nc}{\newcommand}
\nc{\poly}{\operatorname{poly}}
\nc{\polylog}{\operatorname{polylog}} \nc{\Lip}{\operatorname{Lip}}

\nc{\eq}[1]{(\ref{eq:#1})} 
\nc{\eqs}[2]{\eq{#1} and \eq{#2}}

\nc{\eqn}[1]{Eq.~(\ref{eqn:#1})}
\nc{\eqns}[2]{Eqs.~(\ref{eqn:#1}) and (\ref{eqn:#2})}

\newcommand{\secref}[1]{Section~\ref{sec:#1}}

\newcommand{\thmref}[1]{Theorem~\ref{thm:#1}}
\newcommand{\propref}[1]{Proposition~\ref{prop:#1}}

\nc{\region}{\cS\cW}

\begin{document}
\title{Quantum Error Correcting Codes in Eigenstates\\ of Translation-Invariant Spin Chains}

\author{Fernando G.S.L. Brand\~ao}
\affiliation{Department of Physics and Institute for Quantum Information and Matter, California Institute of Technology , Pasadena, CA 91125, USA}
\affiliation{Kavli Institute for Theoretical Physics, UCSB, Santa Barbara, CA 93106, USA}

\author{Elizabeth Crosson}
\affiliation{Department of Physics and Institute for Quantum Information and Matter, California Institute of Technology , Pasadena, CA 91125, USA}
\affiliation{Kavli Institute for Theoretical Physics, UCSB, Santa Barbara, CA 93106, USA}

\author{M.~Burak \c{S}ahino\u{g}lu}
\affiliation{Department of Physics and Institute for Quantum Information and Matter, California Institute of Technology , Pasadena, CA 91125, USA}
\affiliation{Kavli Institute for Theoretical Physics, UCSB, Santa Barbara, CA 93106, USA}

\author{John Bowen}
\affiliation{University of Chicago, Department of Physics, Chicago, IL 60637 USA}

\date{\today}

\begin{abstract}
Quantum error correction was invented to allow for fault-tolerant quantum computation. Systems with topological order turned out to give a natural physical realization of quantum error correcting codes (QECC) in their groundspaces. More recently, in the context of the AdS/CFT correspondence, it has been argued that eigenstates of CFTs with a holographic dual should also form QECCs. These two examples raise the question of how generally eigenstates of many-body models form quantum codes. In this work we establish new connections between quantum chaos and translation-invariance in many-body spin systems, on one hand, and approximate quantum error correcting codes (AQECC), on the other hand. We first observe that quantum chaotic systems exhibiting the Eigenstate Thermalization Hypothesis (ETH) have eigenstates forming approximate quantum error-correcting codes. Then we show that AQECC can be obtained probabilistically from translation-invariant energy eigenstates of every translation-invariant spin chain, including integrable models. Applying this result to 1D classical systems, we describe a method for using local symmetries to construct parent Hamiltonians that embed these codes into the low-energy subspace of gapless 1D quantum spin chains.  As explicit examples we obtain local AQECC in the ground space of the 1D ferromagnetic Heisenberg model and the Motzkin spin chain model with periodic boundary conditions, thereby yielding non-stabilizer codes in the ground space and low energy subspace of physically plausible 1D gapless models.  
\end{abstract}

\maketitle

\section{Introduction}

Quantum error correcting codes (QECC) were originally designed for fault-tolerant quantum computation \cite{shor}. The idea is to cleverly encode the quantum information into entangled states in a way that the information is inaccessible locally. At first sight, it may seem the conditions for quantum error correction are very different from everything we have normally in nature, and that it would take very special engineered quantum systems to realize it. This intuition turned out to be wrong; QECCs appear naturally in the groundspace of topological ordered systems \cite{kitaev2003fault}. This connection has lead to many insights both in the study of quantum error correction \cite{surfaceexp, brown} and of topological order \cite{haah, fu} in the past 20 years. 

In a different direction, in recent years there have been ongoing efforts of connecting the holographic correspondence to quantum error correction. In the  Anti-de Sitter (AdS)/Conformal Field Theory (CFT) correspondence~\cite{maldacena1999large,witten1998anti}, it has been understood to a certain degree that, bulk local operators in AdS are dual to nonlocal operators on the boundary CFT~\cite{HKLL}. Quantum error correction has recently been used~\cite{almheiri2015bulk} for explaining seemingly puzzling facts about this correspondence. It was argued that bulk local operators, reconstructed on the boundary, should commute with boundary local operators only within a certain subspace of the full boundary CFT Hilbert space. Interpreting this subspace as the code subspace of an error correcting code not only clears the apparent puzzles but also gives a new information-theoretic perspective to the AdS/CFT correspondence. Since then quantum error correction has served as a guiding feature for the application of tools from quantum information to the challenge of constructing explicit realizations of AdS/CFT duality~\cite{pastawski2015holographic, harlow2016jerusalem, hayden2016holographic}. Understanding holographic codes from the perspective of the CFT continues to be a major open challenge~\cite{pastawski2017towards, kim2017entanglement}. 

In this Letter, we explore one-dimensional physical systems through the lens of AQECC. Guided by the codes found in the ground space of topologically ordered gapped Hamiltonians and the expectation of good codes in eigenspaces of certain CFTs (motivated by AdS/CFT correspondence), we ask what other physical conditions lead to good quantum codes. First we observe a connection between quantum chaos and quantum error correction, pointing out that the Eigenstate Thermalization Hypothesis (ETH) \cite{srednicki1994chaos} can be interpreted as saying that eigenstates with close-by energies form an AQECC. This observation directly supports the QECC view of the AdS/CFT correspondence, as the CFTs considered there are expected to be chaotic. Then we show that merely translation-invariance of the Hamiltonian already implies that most (translation-invariant) eigenstates in a subextensive energy window of finite energy density form AQECCs.  This general result also applies to integrable models and even to non-interacting Hamiltonians. In some of these cases we show that it is possible to use local symmetries of the states to generate an interacting Hamiltonian that embeds the finite energy eigenstates (i.e., the codespace) of the noninteracting Hamiltonian into the groundspace or low-lying energy subspace of gapless 1D quantum systems. As examples we show how this procedure can give rise to the Heisenberg and Motzkin models.  For these systems we confirm the AQECC performance of the low energy eigenspace by direct calculations, thereby showing that non-stabilizer codes can appear at low energy in physically plausible 1D models. The precise statements about the distance, the dimension of the codespace and the scaling of the error of the AQECC, are given for each case.

\section{Approximate QECC}\label{sec:AQECC)} We start with a brief description of the features of approximate quantum error correction.  For exact quantum error correction, Knill and Laflamme gave a convenient set of necessary and sufficient conditions for a code being able to correct a noisy channel~\cite{KnillLaflamme}. Similar conditions for the approximate case were found by Beny and Oreshkov \cite{beny2010general}, which we now review. We consider $N$ qubits arranged in a line and assume that errors are local. We say that a subspace ${\cal C}$ of a $2^N$-dimensional vector space is a $[[N, k, d, \varepsilon]]$ approximate quantum error correction code (AQECC) if $\text{dim}({\cal C}) = 2^k$ and for every channel $\Lambda$ acting on at most $d$ consecutive qubits, we have
\begin{equation} \label{AQECC}
\min_{\ket{\psi} \in {\cal C}^{\otimes 2}} \max_{{\cal D}}  \bra{\psi}({\cal D} \circ \Lambda \otimes I)(\ket{\psi}\bra{\psi}) \ket{\psi} \geq 1 - \varepsilon,
\end{equation}
where the maximum is over decoding channels ${\cal D}$, and the minimum is over pure entangled states acting on ${\cal C}$ and a reference system (we denote the tensor product space of ${\cal C}$ and the reference by ${\cal C}^{\otimes 2}$ above). In words, this condition states that one can correct, up to error $\varepsilon$, the effect of local noise on at most $d$ qubits. If Eq.~\eqref{AQECC} only works for a particular $\Lambda$, we say the code is $\varepsilon$-correctable under $\Lambda$. 

In this work we find it convenient to consider a set of codewords that span the code space,  $\mathcal{C}= \textrm{span}(\{|\psi_1\rangle, ... , |\psi_{2^k}\rangle\}) \subset \mathbb{C}^{2^N}$, and show that these codewords satisfy an approximate version of the Knill-Laflamme conditions,
\begin{equation}
\langle \psi_i | E | \psi_j\rangle = C_{E} \delta_{ij} + \varepsilon_{ij}. \label{eq:aqec}
\end{equation}
Corollary \ref{easytostrongmapping} of the Appendix shows that if this condition is satisfied, then the error of the code as defined in (\ref{AQECC}) can be bounded as $\varepsilon \leq 2^{2(k+d)}\max_{i,j}\varepsilon_{ij}$.
For many-body systems with $N$ sites, it is natural to seek $\varepsilon \leq \mathcal{O}(N^{-c})$ so that the probability of recovering the logical state converges to 1 quickly with increasing system size.

\section{AQECC from ETH} 
The Eigenstate Thermalization Hypothesis (ETH) states that thermalization in a quantum system takes place already on the level of eigenstates. Given the Hamiltonian $H = \sum_k E_k \ket{E_k}\bra{E_k}$, with $|E_k\rangle$ being energy eigenstates with eigenvalue $E_k$ (ordered as $E_1 \leq E_2 \leq ...$ ), Srednicki proposed the following version of ETH \cite{srednicki1994chaos}: There are constants $c_1, c_2 > 0$ such that for every $E_l, E_k$ in the bulk of the spectrum and for any local observable $O$,
\begin{equation} \label{eq:eth}
|\bra{E_l} O \ket{E_l} - \bra{E_{l+1}} O \ket{E_{l+1}}| \leq \exp(- c_1 N),
\end{equation}
and
\begin{equation} \label{eq:eth2}
|\bra{E_k} O \ket{E_l}| \leq \exp(- c_2 N).
\end{equation}

Indeed Eq.~\eqref{eq:eth} tells us that the energy eigenstates around $\overline{E}$ are locally indistinguishable from each other, and therefore also from the thermal state of the same energy. They ensure that the long-time average of any local observable is thermal. Eq.~(\ref{eq:eth2}), in turn, guarantees that the fluctuations around the long-time average is small.  Comparing the ETH condition Eq.~\eqref{eq:eth} to the AQECC condition Eq.~\eqref{eq:aqec}, we observe that:
\begin{remark}
ETH implies that any region of the spectrum with finite energy density have eigenstates forming approximate error correcting codes.
\end{remark}

Note that the distance of the code is given by the range of locality for which ETH holds in the system. This is expected to vary depending on the model, and can be as large as a constant fraction of the size of the system \cite{GarrisonGrover15}. From Eq.~\eqref{eq:aqec} and Corollary~\ref{easytostrongmapping} of the Appendix, we find that the codes have constant rate, i.e. $k = \Omega(N)$, and exponentially small error. Note that, these are very good codes for highly chaotic systems in which ETH holds for $d$-local observables with $d = \Omega(N)$. 

However, a major drawback is that the codewords are exponentially close to each other in energy, hence it is not clear at all if the Hamiltonian can help with encoding and decoding. One way forward is to split the codewords in energy by sacrificing either the dimension of the codespace or the error of the code. We leave to future work to investigate whether the locality of the Hamiltonian leads to good ways of encoding and decoding in this case.

Notice that ETH codes introduced above are somewhat analogous to random subspace codes (in terms of the parameters achieved)~\cite{hayden2004}. This is no coincidence. One of the ways of understanding quantum chaos is that apart from a few conserved quantities (e.g., energy), the physics of the model mimics the one of a fully random system. Here we give a coding perspective of this view.  

An important application of the observation is in connection to the recent proposal of interpreting some aspects of the AdS/CFT correspondence as an error correcting encoding of the AdS bulk into the boundary CFT~\cite{almheiri2015bulk}. It is expected that holographic CFTs are chaotic and thus satisfy ETH \cite{lashkari2016}. Therefore our observation provides strong evidence in favor of the proposal in Ref.~\cite{almheiri2015bulk}. However, ETH is a claim about eigenstates with finite energy density, whereas the error correcting properties of eigenstates of holographic CFTs are expected to hold even at zero energy. We will partially address this point later in the paper, constructing specific examples of gapless spin chains with AQECC in their low-lying spectrum. The connection of ETH and AQECC that we point out also suggests that such holographic CFTs might be chaotic in an extreme sense of satisfying ETH at all energies. 

\section{AQECC from Translation-Invariance} Although ETH is expected to hold for a large class of systems, its range of validity is still not completely understood. Our next result shows that even just from translation invariance we can already get codes from eigenstates of local models (albeit with worse parameters). Consider a 1D translation invariant Hamiltonian with $N$ sites. 
Let $S_E$ be the set of energy eigenvalues close to $E$: $S_E := \{ E_k :   E_k \in [E - \sqrt{N}, E + \sqrt{N}]  \}$, and define the microcanonical state of energy $E$ as 
\begin{equation}
\tau_{\text{MC(E)}} := \frac{1}{|S_E|} \sum_{k : E_k \in S_{E}} \ket{E_k}\bra{E_k}.
\end{equation}

Note that in one-dimension the correlation length is a function of mean energy $e := E/N$ only, and it is a constant independent of system size when $e$ is too. The choice $2 \sqrt{N}$ for the energy window is arbitrary; all we need is that the associated microcanonical ensemble has finite correlation length, which is true as long as it is subextensive and larger than $\polylog(N)$ \cite{brandao2015equivalence}. 

We prove that:

\begin{thm}\label{thm:main}
Let $H$ be a 1D translation invariant local Hamiltonian and $E$ be such that the microcanonical state at energy $E$ has finite correlation length (independent of system size). Pick $\ket{E_{i_1}}, \ldots, \ket{E_{i_L}}$ uniformly independently at random from $S_{E} := \{ \ket{E_{i}} : E_i \in [E - \sqrt{N}, E + \sqrt{N}] \}$, where $\{ \ket{E_i} \}_{i}$ is a  basis of translation-invariant eigenstates of $H$, and $k := \log(L) = \Omega (\log(N))$. Then with high probability they form an $[[N, k, d, \varepsilon]]$ AQECC with $\varepsilon = O(1/N^{1/8})$ and 
\begin{equation}
d = \min \left( \Omega(\log(N)), \min_{p \neq q \in [L]} |E_{i_{p}} - E_{i_{q}}| - O(\log(N))  \right).
\end{equation} 
\end{thm}

Note that by choosing $k = \delta \log(N)$ for sufficiently small $\delta$, the minimum energy gap will be of order $n^{\Omega(1)}$, and thus the distance of the code is $\Omega(\log(N))$ with high probability. 

The proof in \secref{proofTI} builds upon two results. First, the result of \cite{mori2016weak} establishes a weak version of the eigenstate thermalization hypothesis (ETH) for 1D translation invariant systems (see Lemma \ref{Moriquantitative}): The fraction of the nonthermal energy eigenstates around the microcanonical energy $E$ is exponentially small with the system size $N$. This means that with high probability, randomly chosen codewords do look like the thermal state, and hence are locally indistinguishable. Second, the result from \cite{arad2016connecting} states that eigenstates of general (not necessarily translation-invariant) local Hamiltonians with different energies cannot be ``connected" by local operators, in the sense that the off-diagonal matrix elements of the local operator in energy eigenbasis drop off exonentially with the energy gap (see Lemma \ref{offdiagonalsmall} in \secref{proofTI}). This tells us to choose the codewords sufficiently far apart in energy so that we have the desired distance for the code. 

Translation invariance is crucial in the proof of the results. Technically, it allows us to replace the local observable by an extensibe observable, given by a sum of trnaslations of the original one. Then we can use techniques of large deviation bounds on the measurement of extensive observables in non-critical spin systems to obtain the result. Intuitively, translation invariance guarantees that the information of the codewords is spread to the whole system ``uniformly", and hence cannot be corrupted locally by noise. 

Note that in addition to translation invariance, the only feature of 1D systems we use in the proof is that the microcanonical states at finite energy density always have a finite correlation length. Therefore the theorem generalizes to higher dimensions for eigenstates with finite energy densities (albeit with a worse scale of the error of the code). 

\section{AQECC from the Low-Energy Eigenspace of Gapless Models}

So far we have considered eigenstates at finite energy density. Here we show they are also relevant to the low-lying spectrum of gapless models. We first apply \thmref{main} to noninteracting models, and map the codewords at finite energy eigenstates to low-energy eigenspace of interacting models. We then further analyze the performance of these specific codes by explicitly revealing the working code subspace.

\vspace{0.2 cm}

\noindent \textbf{Classical Models:}
Consider a 1-local Hamiltonian on a system of $N$ qubits,
\begin{equation}
H = \sum_{i=1}^N \frac{1}{2}\left(I - \sigma^z_i\right),\label{eq:1localex}
\end{equation}
which has eigenvalues $0,1,...,N$.  \thmref{main} implies that with high probability a subset of $L$ randomly chosen translation invariant eigenstates of Eq.~\eqref{eq:1localex} with energies in $[\frac{N}{2} -\sqrt{N},\frac{N}{2} + \sqrt{N}]$ will be an AQECC with $\log(N)$ distance. As eigenstates, we can take uniform superpositions of $\sigma^z$-basis states $|\mathbf{s}\rangle$, where $\mathbf{s} = (s_1,...,s_N) \in \{-1,1\}^N$, with a particular magnetization $M(\mathbf{s}) = \sum_{i=1}^N s_i$, 
\begin{equation}
|h_m^N\rangle=  \frac{1}{\sqrt{{N \choose N/2 + m/2}}} \displaystyle\sum_{\mathbf{s} : M(\mathbf{s}) = m} |\mathbf{s} \rangle . \label{eq:Heisenbergstates}
\end{equation}

\noindent \textbf{Mapping to Low-Lying Eigenstates:} Although \thmref{main} only applies to states with finite energy density (when the correlation length of the microcanonical state is finite), it turns out that the excited state AQECC in the example above can be embedded into low energy states of a different local model.  This connection is based on the fact that the permutation symmetric energy eigenstates \eqref{eq:Heisenbergstates} of the spin-1/2 model \eqref{eq:1localex} also span the ground space of the ferromagnetic Heisenberg model,
\begin{equation}
H = -\frac{1}{2}\sum_{j = 1}^N \left(\sigma^x_j \sigma^x_{j+1} + \sigma^y_j \sigma^y_{j+1}+ \sigma^z_j \sigma^z_{j+1}\right).\label{eq:Heisenbergham}
\end{equation}
For ease of notation we consider the version of this model with periodic boundary conditions (PBCs).  In the Appendix we choose codewords with magnetization in the range $(-\sqrt{N},\sqrt{N})$ and show the following proposition by explicit calculation. 
\begin{prop}\label{prop:Heisenberg}
For any $a, b > 0$ with $5 a/2 +  b < 1/2$ the ground space of the spin 1/2 ferromagnetic Heisenberg model with $N$ sites and PBCs contains an $[[N, k, d, \varepsilon]]$ AQECC with $k = a \log N$ , $d = b \log N$, and $\varepsilon = \mathcal{O}\left(\frac{\log^{2} N}{N^{1/2 - 5 a/2 - b}} \right)$.
\end{prop}
Specifically, we prove \propref{Heisenberg} in terms of Eq.\eqref{eq:aqec}. A $d$-local error can change the magnetization by at most $2d$, so for different codewords, i.e. the case $i \neq j$, we have zero error in Eq.\eqref{eq:aqec}. Furthermore, the $d$-body reduced density matrix of different codewords are indistinguishable in the thermodynamic limit, i.e. this gives the error for the cases $i=j$ in Eq.\eqref{eq:aqec}.  
Note that the AQECC parameters achieved in \propref{Heisenberg} are asymptotically equivalent to those in \thmref{main}, though one difference is that in \propref{Heisenberg} the codewords are chosen deterministically.  Finally, we note that the existence of error correcting codes in the ground space of Heisenberg models has been observed before~\cite{ouyang2014permutation, ouyang2016permutation}, although the choices of code words as well as the QEC parameters differ in that work from the ones presented here.  

\vspace{0.2 cm}

Just as finite energy density codes of \eq{1localex} can be embedded in the ground space of the Heisenberg model, one can also consider the spin 1 version of \eq{1localex}, 
\begin{equation}
H = \sum_{i=1}^N \frac{1}{2}\left(I - S^z_i\right) .\label{eq:1localexSpin1}
\end{equation}
The permutation invariant eigenstates of \eq{1localexSpin1} are uniform superpositions of basis states $|\mathbf{w}\rangle$, where $\mathbf{w} = (w_1,...,w_N) \in \{-1,0,1\}^N$, with a particular magnetization $M(\mathbf{w}) = \sum_{i=1}^N w_i$, 
\begin{equation}
|g^N_m\rangle=  \dfrac{1}{|g^N_m|} \displaystyle\sum_{\mathbf{w} : M(\mathbf{w}) = m} |\mathbf{w} \rangle . \label{eq:motzkinstates}
\end{equation}
By \thmref{main} a randomly chosen subset of $L$ states of the form \eq{motzkinstates} with magnetization $m \in (-\sqrt{N},\sqrt{N})$ will with high probability form an AQECC with distance $\Theta(\log N)$.  Just as a finite energy density AQECC of \eq{1localex} was turned into a ground space AQECC of \eq{Heisenbergham}, we seek a parent Hamiltonian which contains the states \eq{motzkinstates} in its ground space.  

Such a parent Hamiltonian can be constructed by using the connection between classical random walks (and more generally reversible Markov chains) and stoquastic frustration free local Hamiltonians~\cite{aharonov2003adiabatic, verstraete2006criticality, bravyi2009complexity}.  The following rules applied to any pair of consecutive basis labels (with periodic boundary conditions) suffice to connect all of the basis states at each energy,
\begin{equation*}
|1, -1\rangle \leftrightarrow |0,0\rangle \quad , \quad |0, 1\rangle \leftrightarrow |1, 0\rangle \quad , \quad |0, -1\rangle \leftrightarrow |-1, 0\rangle. \label{eq:localMoves}
\end{equation*}
These local moves can be adjusted into a local Hamiltonian such that the states constructed as the uniform superposition of basis states of the same energy become the ground states:
\begin{equation}
H = \sum_{j = 1}^N \left(|F\rangle \langle F|_{j,j+1} + |U\rangle \langle U|_{j,j+1} + |D\rangle \langle D|_{j,j+1} \right),\label{eq:motzkindef}
\end{equation}
with $|F\rangle = \frac{1}{\sqrt{2}}\left(|ud\rangle - |00\rangle \right), |U\rangle = \frac{1}{\sqrt{2}}\left(|0u\rangle - |u0\rangle \right)$, $|D\rangle = \frac{1}{\sqrt{2}}\left(|0d\rangle - |d0\rangle \right)$, where the labels $-1,0,1$ is replaced by $d,0,u$. This model is called the spin-1 Motzkin chain with periodic boundary conditions (PBCs)~\cite{bravyi2012criticality,movassagh2015power}.  
Using the well-studied analytical properties ofthe ground states of these models, we prove the following proposition in the Appendix.    
\begin{prop}\label{prop:motzkin}
For any $a, b > 0$ with $5 a/2 +  b < 1/2$ the ground space of the spin 1 Motzkin model on $N$ sites with PBCs contains an $[[N, k, d, \varepsilon]]$ AQECC with $k = a \log N$ , $d = b \log N$, and $\varepsilon = \mathcal{O}\left(\frac{\log^{2} N}{N^{1/2 - 5 a/2 - b}} \right)$.
\end{prop}
The intuitive explanation and the calculations are similar to those for the Heisenberg model. These results also hold for the degenerate Heisenberg and Motzkin chains with open boundary conditions with the restriction that errors are only applied far from the endpoints of the chain.  Finally, we note that it is possible to perturb the model with a local translation invariant field in such a way that that $|g^N_0\rangle$ is the unique ground state, with an inverse polynomial gap to the first excited state~\cite{movassagh2015power}.  With this perturbation the states $|g^N_m\rangle$ gain an energy that increases with the magnitude of $m$, but which vanishes in the thermodynamic limit.  This variant of the Motzkin chain is of interest in the present context because it shows that it is possible for models with a unique ground state to be part of a code space that includes gapless excitations.

\section{Conclusions}

In this Letter we have given new examples of approximate quantum error correction against local noise in the energy eigenstates of physical systems, which goes beyond the well-studied ground states of gapped topologically ordered systems. To be more specific, we have explicitly showed that energy eigenstates packed around some finite energy density eigenstate $E$ of systems exhibiting ETH, and almost all translation invariant finite energy eigenstates of 1D translation invariant local Hamiltonians, construct approximate error correcting codes. We applied the latter result to noninteracting local Hamiltonians to map the finite-energy-density codes to the low-energy subspace of interacting Hamiltonians, eg. Heisenberg model and spin-1 Motzkin chain. We studied the ground states of these models with periodic boundary conditions and further detailed the parameters of the approximate error correcting code that can be found in their low energy.

One can interpret our results from many perspectives. One perspective may be that it is not unusual to find error correcting codes in physical systems; it is indeed a generic phenomena as shown by our results of AQECC from systems with ETH and translation invariance. Another point of view which builds upon the first one is that even though error correcting codes can be found easily in Hamiltonian systems, their varying performance under different types of errors may be a way to characterize different properties of these physical systems. For example, the Motzkin spin-1 model that we analyzed is gapless, however the gap closes as $O(N^{-2})$ on the contrary to $O(N^{-1})$ observed in 1D lattice models whose critical points are effectively described by CFTs. To pursue its potential relevance to AdS/CFT(-like correspondence), one shall follow \cite{almheiri2015bulk,Harlow2017} where certain properties of AdS/CFT such as radial commutativity, subregion duality and Ryu-Takayanagi formula have been matched to operator algebra quantum error correcting codes.

There are numerous other questions one can ask building upon our work. Hence, our results shall best be taken as a first step to elucidate the role of error correcting codes in physical systems, from topological order to ETH, AdS/CFT, and gapless quantum systems. The performance of these codes under specific noise channels must be intimately connected to the physical properties manifested by these systems.

\section{Acknowledgments}

We thank Xi Dong, Tarun Grover,  Nick Hunter-Jones, Robert Koenig, John Preskill for discussions. E.C. is grateful for support provided by the Institute for Quantum Information and Matter, with support of the Gordon and Betty Moore Foundation (GBMF-12500028). M.B.S. acknowledges the support from Simons Qubit fellowship provided by Simons Foundation through It from Qubit collaboration. F.B, E.C. and M.B.S. were supported by an NSF Physics Frontiers Center (NSF Grant PHY-1125565). This research was supported in part by the National Science Foundation under Grant No. NSF PHY-1125915.

\bibliographystyle{unsrt}
\bibliography{main}
\pagebreak
\clearpage
\widetext

\appendix

\section{Approximate Quantum Error Correcting Codes}

Here we give a brief description of the features of approximate quantum error correction that we use. We follow closely \cite{beny2010general}. We say that a subspace ${\cal C}$ of a $2^N$-dimensional vector space of $N$ qubits arranged in a line is a $[[N, k, d, \varepsilon]]$ approximate quantum error correction code (AQECC) against $d$-local errors if $\text{dim}({\cal C}) = 2^k$ and for every channel $\Lambda$ acting on at most $d$ consecutive qubits, we have
\begin{equation} \label{AQECC-appendix}
	\min_{\ket{\psi} \in {\cal C}^{\otimes 2}} \max_{{\cal D}}  \bra{\psi}({\cal D} \circ {\cal{N}} \otimes I)(\ket{\psi}\bra{\psi}) \ket{\psi} \geq 1 - \varepsilon,
\end{equation}
where the maximum is over decoding channels ${\cal D}$, and the minimum over pure entangled states acting on ${\cal C}$ and a reference system (which altogether we denoted by ${\cal C}^{\otimes 2}$ above). In words, the condition above says that one can correct, up to error $\varepsilon$, the effect of any noise on at most $d$ neighboring qubits. If \ref{AQECC-appendix} only works for a particular noise channel $\cal N$, we say that the code is $\varepsilon$-correctable under $\cal N$. 

For exact quantum error correction, Knill and Laflamme gave a convenient set of necessary and sufficient conditions for a code being able to correct a noisy channel ${\cal N}(X) = \sum_{k} E_k X E_{k}^{ \cal y}$. Similar conditions for the approximate case were found by Beny and Oreshkov \cite{beny2010general}. 

For two channels ${\cal N}$ and ${\cal M}$, Let $d({\cal N}, {\cal M}) = \sqrt{1 - F({\cal N}, {\cal M})}$ be the Bures metric, where the fidelity of the two channels $F({\cal N}, {\cal M})$ is defined as follows:
\begin{equation}
F({\cal N}, {\cal M}) := \max_{\ket{\psi}} F( I \otimes {\cal N}(\ket{\psi}\bra{\psi}), I \otimes {\cal M}(\ket{\psi}\bra{\psi})),
\end{equation}
with the maximiztion over all bipartite states $\ket{\psi}$ of the input of the channel and a vector space isomorphic to it. 

Then  we have:
\begin{prop} \label{BOKnillLaflammecond}
	[Corollary 2 of \cite{beny2010general}] A code defined by the projector $P$ is $\varepsilon$-correctable under a noise channel $\Lambda$ if and only if
	\begin{equation} \label{approximateKL}
		PE_{i}^{\cal y} E_{j} P = \lambda_{ij} P + P B_{ij}P, 
	\end{equation}
	where $\lambda_{ij}$ are the components of a density operator, and $d({\cal N} + {\cal B}, {\cal N}) \leq \varepsilon$, where ${\cal N}(\rho) = \sum_{i, j} \lambda_{ij} \tr(\rho) \ket{i}\bra{j}$ and $({\cal N} + {\cal B})(\rho) = {\cal N}(\rho) + \sum_{i, j} \tr(\rho B_{ij})\ket{i}\bra{j}$.
\end{prop}

In the proposition the projector $P$ is the projector onto the support of the subspace ${\cal C}$ defining the code. An easy consequence of Proposition \ref{BOKnillLaflammecond} is the following:
\begin{cor} \label{easytostrongmapping}
Let  $\{ |\psi_1\rangle, ... , |\psi_{2^k} \rangle \}$ be an orthogonal set of states in $\mathbb{C}^{2^N}$ such that for all $i,j$ and any $d$-local operator $E$,
\begin{equation}
\langle \psi_i | E | \psi_j\rangle = C_{E} \delta_{ij} + \varepsilon_{ij}, \label{eq:aqec2}
\end{equation}
with $C_E$ a constant (only depending on $E$). Then ${\cal C} := \text{span}\{ |\psi_1\rangle, ... , |\psi_{2^k} \rangle \}$ forms a $[[N, k, d, 2^{d+2k} \max_{i, j} \varepsilon_{ij}^{1/2}]]$ AQECC.
\end{cor}
\begin{proof}
Let $P := \sum_{l=1}^{2^k} | \psi_l \rangle \langle \psi_l |$.  We find Eq. (\ref{approximateKL}) to be true with $\lambda_{ij} = \langle \psi_1 | E_i E_j | \psi_1 \rangle$ and
\begin{equation}
B_{ij} = \sum_{k \neq l} \langle \psi_l | E_i E_j | \psi_k \rangle | \psi_l \rangle \langle \psi_k | + \sum_{k} \left(\langle \psi_k  | E_i E_j  |  \psi_k \rangle -  \langle \psi_1  | E_i E_j  | \psi_1 \rangle \right)| \psi_k \rangle \langle \psi_k |.  
\end{equation}

We now note two facts: (1) the Bures metric is upper bounded by the trace norm; and (2) for every two channels $\max_{\psi \in {\cal C}^{\otimes 2}} \Vert I \otimes {\cal N}(\ket{\psi}\bra{\psi}), I \otimes {\cal M}(\ket{\psi}\bra{\psi} \Vert_1 \leq 2^{k} \max_{\psi \in {\cal C}} \Vert {\cal N}(\ket{\psi}\bra{\psi}), {\cal M}(\ket{\psi}\bra{\psi}) \Vert_1$  (see Lemma 23 of \cite{hayden2012weak}). Then

\begin{equation}
d({\cal N} + {\cal B}, {\cal N}) \leq 2^{k} \Vert {\cal B} \Vert_1^{1/2} \leq 2^{k} \left( \sum_{i, j}^{2^d} \Vert B_{ij} \Vert_1 \right)^{1/2} \leq 2^{d+k} \max_{i, j} \Vert B_{ij} \Vert_1^{1/2}. 
\end{equation}

From Eq. (\ref{eq:aqec2}) we can bound the square of the latter as follows
\begin{equation}
\Vert B_{ij} \Vert_1 \leq \sum_{k \neq l}^{2^k}  \varepsilon_{kl} + \sum_{k=1}^{2^k} \varepsilon_{kk} \leq 2^{2k} \varepsilon.  
\end{equation}
\end{proof}

\section{Codes from Translation-Invariance} \label{sec:proofTI}

Here we give a proof of Theorem~\ref{thm:main}.

\subsubsection{Diagonal Elements are Close}

We say a state $\rho$ on a finite dimensional lattice has correlation length $\xi$ if
\begin{equation}
\max_{X, Z} \frac{\tr(\rho X \otimes Z) - \tr(\rho X) \tr(\rho Z)}{ \Vert X \Vert \Vert Z \Vert} \leq \exp(- \text{dist}(X, Z)/\xi),
\end{equation}
with the maximimization over all Hermitian matrices $X, Z$. 

The next lemma, due to Anshu \cite{anshu2016concentration}, gives a large deviation principle for the measurement of the energy, according to a local Hamiltonian $H$, on a state with a finite correlation length. 

\begin{lem}  \label{largedeviation}
[Theorem 1.1 of \cite{anshu2016concentration}] Let $\rho$ be a quantum state with correlation length $\xi$ and $\left< H \right >_{\rho} := \tr(\rho H)$ be the average energy of $\rho$. Let $\Pi_{\geq f}$ be the projection onto the eigenspace of $H$ with eigenvalues $\geq f$. 

For $a \geq \sqrt{  \frac{\exp(c_1 D \log(k))}{N \xi}  }$ it holds that
\begin{equation}
\tr( \rho \Pi_{\geq \left <  H  \right >_{\rho} + Na} ) \leq  e^{\frac{2 D k}{\xi}} e^{-  c_2 \frac{(N a^2 \xi)^{\frac{1}{D+1}}}{D \xi}}.
\end{equation}
for an universal constants $c_1, c_2$.
\end{lem}

Given a Hamiltonian $H$ with spectral decomposition $H = \sum_l E_l \ket{E_l} \bra{E_l}$, let $S_E$ be the the set of eigenvalues close to $E$: 
\begin{equation} \label{S}
S_E := \{ E_k :   E_k \in [E - \sqrt{N}, E + \sqrt{N}]  \}.
\end{equation}
Define the microcanonical state of energy $E$ as 
\begin{equation}
\tau_{E} := \frac{1}{|S_E|} \sum_{k : E_k \in S_{E}} \ket{E_k}\bra{E_k}.
\end{equation}
We note that in one dimension the correlation length is a function of mean energy $e := E/N$ only, and it is a constant independent of system size when $e$ is a contant as well. The choice $2 \sqrt{N}$ for the energy window is arbitrary. All we need is that the associated microcanonical ensemble has finite correlation length, which is true as long as it is subextensive and larger than $\polylog(N)$ \cite{brandao2015equivalence}. 

The following is a quantitative version of the main result of \cite{mori2016weak}, which established a weak version of ETH (only concerned with diagonal elements and only applying to most eigenstates). Using the large deviation principle of Lemma \ref{largedeviation} we can give a finite version of it, with error bounds (in contrast, the result of \cite{mori2016weak} concerns asymptotics). It is here that the assumption of having translation-invariant eigenstates is used.

\begin{prop}   \label{Moriquantitative}

There is a constant $0 < \alpha < 1/2$ such that the following holds. Let $H$ be a 1D local Hamiltonian on $N$ qubits and $O$ be an observable acting non-trivially only on a connected region of length $N^{\alpha}$. Then for any $\delta \geq 1/N$,
\begin{equation}
\Pr_{\ket{E_k} \in S_{E}}  \left( \bra{E_k} O \ket{E_k}  \geq  \delta   \right)  \leq \exp\left( - c \delta N^{1/2} / \xi^{3/2}  \right),
\end{equation}
for a universal constant $c$, with $\xi$ the correlation length of the microcanonical state $\tau_{E}$ of energy $E$. 
\end{prop}

\begin{proof}
For a local observable $O$, define $O' := O -  \left <   O  \right >_{\text{MC}(E)} I$, with the average over the microcanonical state, i.e. $\left <   O  \right >_{\text{MC}(E)} := \tr( \tau_{E}  O)$. Define  $\overline{O'} := \sum_{i} T_i(O')$ with $T_i$ denoting a translation by $i$ sites. Following \cite{mori2016weak}, for any $\lambda > 0$ we have:
\begin{eqnarray}
\Pr_{\ket{E_k} \in S_{E}}  \left( \bra{E_k} O \ket{E_k} \geq \left <   O  \right >_{ \tau_{E}    } +  \delta   \right) &=& \Pr_{\ket{E_k} \in S_{E}}  \left( e^{\lambda \bra{E_k} \overline{O'} \ket{E_k}} \geq e^{\lambda \delta N}  \right) \nonumber \\
&\leq& e^{- \lambda \delta N} \mathbb{E}_{\ket{E_k} \in S_{E}}   \left(  e^{\lambda \bra{E_k} \overline{O'} \ket{E_k}}   \right) \nonumber \\
&\leq& e^{- \lambda \delta N} \mathbb{E}_{\ket{E_k} \in S_{E}}   \left(  \bra{E_k} e^{\lambda  \overline{O'}}\ket{E_k}   \right) \nonumber \\ 
&=& e^{- \lambda \delta N}  \left <   e^{\lambda  \overline{O'}}  \right >_{\text{MC}(E)}, \nonumber
\end{eqnarray}
where the first inequality follows from Markov's inequality and the second from the inequality: $e^{\bra{\psi} X \ket{\psi}} \leq \bra{\psi} e^X \ket{\psi}$, valid for any Hermitian matrix $X$ and state $\ket{\psi}$.  

The key step of the proof is the first line of the equation above where we used that $\bra{E_k}O \ket{E_k} = \bra{E_k} \overline{O'} \ket{E_k}/N$. This relation holds because the eigenstates are translation invariant. Therefore, we can replace the expectation of a local observable by the expectation value of an extensive observable, which allows us to bring the well-developed machinery of large deviation bounds for spins systems, which we now employ.

Let $\overline{O'} = \sum_{i} o_i P_i$ be the spectral decomposition of $\overline{O'}$. We have
\begin{eqnarray}
&&  e^{- \lambda \delta N}  \left <   e^{\lambda  \overline{O'}}  \right >_{\text{MC}(E)}  =  e^{- \lambda \delta N}  \sum_{i : o_i \leq \delta N/2} e^{\lambda o_i} \tr( \tau_{E}  P_i) +  e^{- \lambda \delta N}  \sum_{i : o_i > \delta N/2} e^{\lambda o_i} \tr( \tau_{E}  P_i).  \nonumber
\end{eqnarray}

We can upper bound the first term as follows:
\begin{equation} \label{firstbound}
e^{- \lambda \delta N} \sum_{i : o_i  \leq \delta N/2} e^{\lambda o_i} \tr( \tau_{E}  P_i) \leq e^{- \lambda \delta N/2}.
\end{equation}

For the second term, we have

\begin{eqnarray}
\sum_{i : o_i > \delta n/2} e^{\lambda o_i} \tr( \tau_{E}   P_i)  &=&  \sum_{j=0}^{4N} \left( \sum_{i :  \delta N/2 + (j + 1)  \geq o_i > \delta N/2 + j} e^{\lambda o_i} \tr( \tau_{E}   P_{i })  \right)  \nonumber \\
&\leq&  \sum_{j=0}^{4N}  e^{\lambda(\delta N / 2 + j + 1)} \tr( \tau_{E}   P_{\geq \delta N/2 + j}) := X.
\end{eqnarray}
with $P_{\geq \delta N/2 + j} = \sum_{j : o_j \geq  \delta N/2 + j} P_j$. 

Using Lemma \ref{largedeviation},
\begin{equation}
X \leq  \sum_{j=0}^{4N}  e^{\lambda(\delta N / 2 + j + 1)}  e^{\frac{2 N^{\alpha}}{\xi}} e^{ - c  \left( N \left( \frac{\delta}{2} + \frac{k}{N}   \right)^2 \xi  \right)^{1/2} } .
\end{equation}
Choosing
\begin{equation}
\lambda :=  O(\xi^{3/2} N^{-1/2})
\end{equation}
we find 
\begin{equation} \label{boundX}
X \leq e^{\frac{2 N^{\alpha}}{\xi}}  \exp \left( - c' \delta N^{1/2} / \xi^{3/2}  \right).
\end{equation}
Eqs. (\ref{boundX}) and (\ref{firstbound}) gives the statement. 
\end{proof}

A direct consequence of the proposition above is the following:

\begin{cor}  \label{diagonalequal}
Let $H$ be a 1D local Hamiltonian on $N$ qubits and $Z$ be a connected region with less than $O(\log(N))$ sites. Then 
\begin{equation}
\Pr_{\ket{E_k}, \ket{E_l} \in S_{E}}  \left(   \Vert  \tr_{\backslash Z}(\ket{E_k}\bra{E_k}) -  \tr_{\backslash Z}(\ket{E_l}\bra{E_l}) \Vert_1 \geq \delta \right)  \leq  \exp \left( - c \delta N^{1/2} / \xi^{3/2}  \right)
\end{equation}
for a constant $c$, with $\xi$ the correlation length of the microcanonical state of energy $E$. We denote by $\tr_{\backslash Z}$ the partial trace over the complement of $Z$. 
\end{cor}

\begin{proof}
Proposition \ref{Moriquantitative} gives that for a fixed Hermitian matrix $O$ with $\Vert O \Vert = 1$ over $d = \eps \log(N)$ sites:
\begin{equation}
\Pr_{\ket{E_k}, \ket{E_l} \in S_{E}}  \left(   \tr(O(\tr_{\backslash Z}(\ket{E_k}\bra{E_k}) -  \tr_{\backslash Z}(\ket{E_l}\bra{E_l}))) \geq \delta/2 \right)  \leq  \exp \left( - c \delta N^{1/2} / \xi^{3/2}  \right).
\end{equation}
Consider a $\delta/2$-net ${\cal N}_{\delta}$  over the set of all Hermitian matrices of unit operator norm. We have $|{\cal N}_{\delta}| \leq (1/\delta)^{O(2^{2d})}$. Using the union bound
\begin{eqnarray}
&& \Pr_{\ket{E_k}, \ket{E_l} \in S_{E}}  \left(    \Vert  \tr_{\backslash Z}(\ket{E_k}\bra{E_k}) -  \tr_{\backslash Z}(\ket{E_l}\bra{E_l}) \Vert_1 \geq \delta/2 \right) \nonumber \\ &\leq& \Pr_{\ket{E_k}, \ket{E_l} \in S_{E}}  \left(       \max_{O \in {\cal N}_{\delta}} \tr(O(\tr_{\backslash Z}(\ket{E_k}\bra{E_k}) -  \tr_{\backslash Z}(\ket{E_l}\bra{E_l}))) \geq \delta \right) \nonumber \\
&\leq& |{\cal N}_{\delta} | Pr_{\ket{E_k}, \ket{E_l} \in S_{E}}  \left(       \max_{O \in {\cal N}_{\delta}} \tr(O(\tr_{\backslash Z}(\ket{E_k}\bra{E_k}) -  \tr_{\backslash Z}(\ket{E_l}\bra{E_l}))) \geq \delta \right) \nonumber \\ 
&\leq &|{\cal N}_{\delta}| \exp \left( - c \delta N^{1/2} / \xi^{3/2}  \right) \leq \exp \left( - c' \delta N^{1/2} / \xi^{3/2}  \right),
\end{eqnarray}
for $\eps > 0$ a constant sufficiently small.  

\end{proof}

\subsubsection{Off-Diagonal Elements are Small}

The next lemma, from Arad, Kuwahara and Landau \cite{arad2016connecting} (and attributed to Hastings), shows that for a local Hamiltonian, eigenstates well separated in energy are not connected by local operators. Its proof uses similar ideas to the proof of the Lieb-Robinson bound:

\begin{lem}  \label{offdiagonalsmall}  
[Theorem 2.1 of \cite{arad2016connecting}] Let $\Pi_{[\varepsilon', \infty]}$ and $\Pi_{[0, \varepsilon]}$ be projectors onto the subspaces of energies of $H$ that are $\geq \varepsilon'$ and $\leq \varepsilon$, respectively. For an operator $O$, let $E_O$ be a subset of interactions terms such that $[H, O] = \sum_{X \in E_O} [h_X, O]$, and let $R := \sum_{X \in E_O} \Vert h_X \Vert$. Then
\begin{equation}
\Vert  \Pi_{[\varepsilon', \infty]} O  \Pi_{[0, \varepsilon]}  \Vert \leq \Vert O \Vert . e^{- \lambda (\varepsilon' - \varepsilon - 2R)},
\end{equation}
with $\lambda := (2gk)^{-1}$ with $k$ the locality of $H$ and $g$ an upper bound on the number of local terms involving each particle.
\end{lem}

The result has a straightforward corollary, which we state for future use:
\begin{cor} \label{smallof2}
Let $\ket{E_{k_1}}$ and $\ket{E_{k_2}}$ be two eigenstates of a local Hamiltonian and $Z$ a connected region. Then
\begin{equation}
\Vert \tr_{\backslash Z}\left( \ket{E_{k_2}}\bra{E_{k_1}} \right)  \Vert_1 \leq e^{- \lambda (|E_{k_1} - E_{k_2}| - 2|Z|)}
\end{equation}
with $Z$ the size of $Z$ and $\lambda > 0$ a universal constant. 
\end{cor}

\subsubsection{Proof of Theorem \ref{thm:main}}

\textbf{Theorem~\ref{thm:main}.}\textit{
Let $H$ be a 1D translation invariant local Hamiltonian and $E$ be such that the microcanonical state at energy $E$ has finite correlation length (independent of system size). Pick $\ket{E_{i_1}}, \ldots, \ket{E_{i_L}}$ uniformly independently at random from $S_{E} := \{ \ket{E_{i}} : E_i \in [E - \sqrt{N}, E + \sqrt{N}] \}$, where $\{ \ket{E_i} \}_{i}$ is a  basis of translation-invariant eigenstates of $H$, and $k := \log(L) = \Omega (\log(N))$. Then with high probability they form an $[[N, k, d, \varepsilon]]$ AQECC with $\varepsilon = O(1/N)^{1/8}$ and 
\begin{equation}
d = \min \left( \Omega(\log(N)), \min_{p \neq q \in [L]} |E_{i_{p}} - E_{i_{q}}| - O(\log(N))  \right).
\end{equation} 
}

\begin{proof}
The theorem is a consequence of Corollary \ref{diagonalequal}, Corollary \ref{smallof2} and Corollary \ref{easytostrongmapping}. Indeed, the union bound and Corollary \ref{diagonalequal} show that with high probability, for every $p \neq q \in [L]$ and $Z$ with $|Z| = O(\log(N))$
\begin{equation}
 \Vert  \tr_{\backslash Z}(\ket{E_{i_p}}\bra{E_{i_p}}) -  \tr_{\backslash Z}(\ket{E_{i_q}}\bra{E_{i_q}}) \Vert_1 \leq N^{-\frac{2}{5}}.  
\end{equation}

Corollary \ref{smallof2}, in turn, gives that for every $p \neq q \in [L]$ and $Z$ with $|Z| = O(\log(N))$
\begin{equation}
\Vert \tr_{\backslash Z}\left( \ket{E_{i_p}}\bra{E_{i_q}} \right)  \Vert_1 \leq e^{- \lambda (|E_{i_p} - E_{i_q}| - \Omega(\log(N)))}.
\end{equation}
These two conditions and Corollary \ref{easytostrongmapping} allows us to bound the error of the code as
\begin{equation}
\varepsilon \leq 2^{2k + d} (e^{- \lambda (|E_{i_p} - E_{i_q}| - \Omega(\log(N)))} + N^{-\frac{2}{5}})^{1/2} \leq O(1/N)^{1/4},
\end{equation}
choosing the several constants appropriately.
\end{proof}

\vspace{0.2 cm}

\noindent \textbf{Observation 1:}  One drawback of the theorem is that if the minimum energy gap is less than $\Omega(\log(N))$, then the distance is zero. With high probability this will not be the case (since the energy window is $2 \sqrt{N}$, we pick $O(\log(N))$ elements uniformly at random, and the energy distribution of eigenvalues of a random model is normal \cite{brandao2015equivalence}). However, if we want to make sure that this bad case will not happen, we can consider a variant of the theorem in which we consider energies $[E, E+2\sqrt{N}, \ldots, E+L\sqrt{N}]$ and pick each state uniformly from $S_{E + 2j \sqrt{N}}$ for $j \in [L]$. 

\vspace{0.2 cm}

\noindent \textbf{Observation 2:} The only feature we used of being in one dimension is that the microcanonical states at finite energy density always have a finite correlation length. Therefore the theorem generalizes to higher dimensions for eigenstates with finite energy densities (albeit with a worse scale of the error of the code). 

\section{Parent Hamiltonians from local symmetries}
Let $H_C$ be a 1D translation invariant classical Hamiltonian with periodic boundary conditions acting on $N$ qudits with local dimension $D$.  The statement that $H_C$ is classical means that there is some tensor product basis $B = \{|s_1,...,s_N\rangle\}$, with $s_i \in \{1,...,D\}$ for each $i$, such that we can express  $H_C$ as
$$
H_C = \sum_{\mathbf{s} \in B}H_C(\mathbf{s}) |\mathbf{s}\rangle \langle \mathbf{s}|
$$
Let $G$ be a locally generated group of symmetries of $H_C$.  Since $H_C$ is translation invariant these generators are described by a set of $k$-local invertible linear maps $r_1,...,r_p$, with $r_i : \mathbb{C}^{D^k} \rightarrow \mathbb{C}^{D^k}$ for each $i = 1,...,p$, and their translations.  The action of $r_i$ on sites $j,...,j+k$ is expressed by $|s_1...s_j ...s_{j+k}...s_N\rangle \mapsto |s_1...r_i(s_j...s_{j+k})...s_N\rangle$.  Since $r_i$ describes a symmetry of $H_C$ it follows that
$$
H_C|s_1...s_j ...s_{j+k}...s_N\rangle = H_C|s_1...r_i(s_j...s_{j+k})...s_N\rangle
$$
for all $i,j,k$ and for all $\mathbf{s} \in B$.  Furthermore, for each $|\mathbf{s}\rangle$ it follows that all of the states in the orbit $G(\mathbf{s}) = \{|g(\mathbf{s})\rangle : g \in G\}$ also have the same energy with respect to $H_C$.     These orbits partition $B$ into subsets $B_1,...,B_m$.

For each local symmetry generator $r_i$ and each site $j \in \{1,...,N\}$ we can define a local projector
\begin{align}
\Pi^r_{j} = \frac{1}{2}\sum_{s_j,...,s_{j+k}}\bigl(|s_j...s_{j+k}\rangle \langle s_j... s_{j+k}|+|r(s_j...s_{j+k})\rangle \langle r(s_j... s_{j+k})|\nonumber \\
-|r(s_j...s_{j+k})\rangle \langle s_j... s_{j+k}|-|s_j...s_{j+k}\rangle \langle r(s_j... s_{j+k})|\bigr),
\end{align}
and if $H$ is the Hamiltonian defined by the sum of all these projectors 
$$
H = \sum_{j = 1}^N \sum_{r = 1}^p \Pi^{r_i}_j.
$$
then the ground space of $H$ is $m$-fold degenerate, and it is spanned by states which are uniform superpositions of the states in each orbit,
\begin{equation}
|\psi_m \rangle = \frac{1}{|B_m|} \sum_{|\mathbf{s}\rangle \in B_m} |\mathbf{s}\rangle.
\end{equation}

\section{Heisenberg spin chain AQECC}

The Hamiltonian of the spin 1/2 quantum Heisenberg chain on $N$ sites with periodic boundary conditions is
\begin{equation}
H = -\frac{1}{2}\sum_{j = 1}^N \left(\sigma^x_j \sigma^x_{j+1} + \sigma^y_j \sigma^y_{j+1}+ \sigma^z_j \sigma^z_{j+1}\right).\label{eq:Heisenbergham-appendix}
\end{equation}
To study this system we use the $\sigma^z$-basis consisting of states $|\mathbf{s}\rangle$, where $\mathbf{s}= (s_1, s_2, \ldots, s_N) \in \{-1,1\}^N$ and each $s_j \in \{-1,1\}$ labels the eigenvalue of $\sigma^z_j$ .  Let $M = \sum_{j = 1}^N \sigma^z_j$ be the magnetization operator and define $M(\mathbf{s}) = \langle \mathbf{s} | M | \mathbf{s}\rangle = \sum_{i = 1}^N s_i$. Since $[M,H] = 0$ the model \eq{Heisenbergham-appendix} has an $N$-fold degenerate ground space, with ground states that are labeld by magnetization values $m \in \{-N, -N+2, -N+4, \ldots, N-2, N\}$.  We denote these ground states by $|h_m^N\rangle$, where $m$ is the magnetization and the system size $N$ is made explicit because we will also consider ground states of the Heisenberg chain on connected subsets of the $N$ sites.  As is well known from the exact solution of \eq{Heisenbergham-appendix} the state $|h_m^N\rangle$ can be expressed in terms of the $\sigma^z$-basis states as
\begin{equation}
|h_m^N\rangle=\dfrac{1}{\sqrt{{N \choose {N/2 + m/2}}}} \displaystyle\sum_{\mathbf{s}: M(\mathbf{s})=m} |\mathbf{s} \rangle .\label{eq:(m,N)}
\end{equation}
As part of the verification of the error correcting conditions we will use expressions for the $d$-body connected reduced density matrices of these ground states. By taking $d$ to always be asymptotically smaller than $N$ we can express the Schmidt decomposition of \eq{(m,N)} along the cut between the $d$ spins where the error acts nontrivially and its compliment,
\begin{equation}\label{Heisenberg-groundstates}
|h_m^N\rangle= \displaystyle\sum^{d}_{r=-d} |h_r^d\rangle |h_{m-r}^{N-d}) \rangle \left[\dfrac{{d \choose d/2 + r/2}{N-d \choose N/2 - d/2 + m/2 - r/2}}{{N \choose N/2 + m/2}}\right]^{1/2}.
\end{equation}
It follows from (\ref{Heisenberg-groundstates}) that the reduced density matrix on $d$ sites that $\rho_{d}(m,N) = \tr_{N-d}\left(|h_m^N\rangle\langle h_m^N|\right)$ (note that $\tr_{N -d}$ refers to the trace over the $(N-d)$-neighboring sites) is given by
\begin{equation}
\rho_d(m,N)= \displaystyle\sum^{d}_{r-d} |h_r^d\rangle \langle h_{r}^d| \left [ \dfrac{{d \choose d/2 + r/2}{N-d \choose N/2 - d/2 + m/2 - r/2}}{{N \choose N/2 + m/2}} \right]. \label{eq:dbodyrdm}
\end{equation}


\subsubsection{Proof of Proposition~\ref{prop:Heisenberg}}
\begin{proof}

We separate the approximate error correction condition into two parts. First is the nonexistence of a $d$-body operator that maps different codewords to each other. More precisely, this corresponds to the case $i \neq j$ in Eq.~\eqref{eq:aqec}. Note that $\langle h_m^N|E |h_{m'}^N\rangle = 0$ by construction whenever $m$ and $m'$ are distinct codewords, since $|m - m'| > 2d$ and $E$ is the error operator supported on a connected region of $d$-sites, so it can change the magnetization by at most $2d$.  Noting that a $d$-local operator can change the magnetization of $\mathbf{s}$ by at most $2d$, define  $I = 2 d + 1$ and define the code space to be
\begin{equation}
\mathcal{C} = \textrm{Span}\left(\left\{|h_m^N\rangle:m=-m_{\max}, -m_{\max} + I, \ldots, m_{\max} - I, m_{\max}  \right\}\right).\label{eq:Heisenbergcodespace}
\end{equation} 
where $m_{\max}$ is a multiple of $I$ to be chosen later.  The second approximate error correction condition is the local indistinguishability of the $d$-body reduced density matrices of the codewords. More precisely, this corresponds to the case $i =j$ in Eq.~\eqref{eq:aqec}. Below we show that an arbitrary observable has approximately the same expectation value for all of the codewords by computing the trace distance between the local reduced density matrices $\rho_{d}(m,N) = \tr_{N-d}\left(|h_m^N\rangle\langle h_m^N|\right), \rho_{d}(m',N) = \tr_{N-d}\left(|h_{m'}^N\rangle\langle h_{m'}^N| \right)$ on the support of the error (note that $\tr_{N -d}$ refers to the trace over the $N-d$ sites where the error acts trivially. The error is assumed to act on $d$-neighboring sites).  Using the $d$-body reduced density matrices \eq{dbodyrdm} this distance is
\begin{equation}
 \|\rho_d(m,N) - \rho_d(m',N)\|_1 = \displaystyle\sum^{d}_{r=-d} {d \choose d/2+r/2}\left|\dfrac{{N-d \choose N/2-d/2+m/2-r/2}}{{N \choose N/2+m/2}}-\dfrac{{N-d \choose N/2-d/2+m'/2-r/2}}{{N \choose N/2+m'/2}}\right|.\label{eq:Heisenbergdiagonalerror}
\end{equation}
Since $\max_{m,m'} |m - m'| \leq 2m_{\max}$ and $0\ll d \ll 2m_{\max} \ll N^{1/2}$ we can apply the asymptotic approximation for the binomial coefficients inside the absolute value, ${a \choose a/2 + b/2} \approx \dfrac{2^{a+1}}{\sqrt{2 \pi a}}e^{-b^2/2a}$ to approximate \eq{Heisenbergdiagonalerror} by
\begin{equation}
\displaystyle\sum^{d}_{r=-d} {d \choose d/2+r/2} \dfrac{2^{-d}\sqrt{N}}{\sqrt{N-d}} \left| \exp \left[\frac{1}{2}\left(\frac{m^2}{N} - \frac{(m-(d+r))^2}{N-d}\right)\right] - \exp \left[\frac{1}{2}\left(\frac{m'^2}{N} - \frac{(m'-(d+r))^2}{N-d}\right)\right] \right|. \label{eq:Heisenbergdiagonalerrorintermediate}
\end{equation}
Note that the asymptotic expression does not apply to ${d \choose d/2+r/2}$ for all values of $r$, so instead we can use ${d \choose d/2+r/2} \leq 2^d  / \sqrt{d}$ to obtain an upper bound on \eq{Heisenbergdiagonalerrorintermediate},
\begin{equation}
\displaystyle\sum^{d}_{r=-d} \dfrac{\sqrt{N}}{\sqrt{d} \sqrt{N-d}} \left| \exp \left[\frac{1}{2}\left(\frac{m^2}{N} - \frac{(m-(d+r))^2}{N-d}\right)\right] - \exp \left[\frac{1}{2}\left(\frac{m'^2}{N} - \frac{(m'-(d+r))^2}{N-d}\right)\right] \right|.
\end{equation}
Applying a Taylor series expansion (again using $0\ll d \ll 2m_{\max} \ll N^{1/2}$) to the term inside the absolute value,
\begin{align}
 \left|\frac{m^2 - m'^2}{N} - \frac{(m-(d+r))^2 - (m'-(d+r))^2}{N-d}\right| & = \left| \frac{2(m-m')(d+r)}{N-d} - \frac{d(m^2-m'^2)}{N(N-d)}\right|\\
 & \leq \left| \frac{2(m-m')(d+r)}{N-d}\right| +\left| \frac{d(m^2-m'^2)}{N(N-d)}\right| ,\label{eq:Heisenbergintermediateexpression}
\end{align}

and we see that the second term in Eq.~\eqref{eq:Heisenbergintermediateexpression} has a subleading contribution to the scaling of the error, so
\begin{align}
 \|\rho_d(m,N) - \rho_d(m',N)\|_1 = \mathcal{O} \left( \frac{\sqrt{N}d^{3/2}|m-m'|}{(N-d)^{3/2}} \right) =\mathcal{O}\left(\frac{d^{3/2}|m - m'|}{N}\right) = \mathcal{O}\left(\frac{d^{3/2}m_{\max}}{N}\right) \label{eq:scalingOfHeisenbergError}.
\end{align}

To encode $k = \lfloor \log_2 \textrm{dim}(\mathcal{C}) \rfloor$ logical qubits it suffices to take $m_{\max} = d \cdot 2^k$.  Therefore for any $a, b > 0$ with $5 a/2 +  b < 1/2$ the ground space of the spin 1/2 Heisenberg model contains an $[[N, k, d, \varepsilon]]$ AQECC with $k = a \log N$ , $d = b \log N$, and $\varepsilon = \mathcal{O}\left(\frac{\log^{2} N}{N^{1/2 - 5 a/2 - b}} \right)$.
\end{proof}

\section{Motzkin spin chain AQECC}\label{sec:Motzkin}

\subsubsection*{Proof of Proposition~\ref{prop:motzkin}}
\begin{proof}
Let $H$ be the Hamiltonian of the spin 1 Motzkin model on $N$ sites with periodic boundary conditions.  For a $S^z$ basis state $|\mathbf{w}\rangle$ with $\mathbf{w} = (w_1,...,w_N) \in \{-1,0,1\}^N$, define the magnetization $M(\mathbf{w})$ of $|\mathbf{w}\rangle$ by $M(\mathbf{w}) = \sum_{i=1}^N w_i$.   For each $m$ define $|g_m^N\rangle$ to be the normalized superposition of all basis states with magnetization $m$,
$$
|g_m^N\rangle = \frac{1}{|g^N_m|}\sum_{\mathbf{w} : M(\mathbf{w}) = m} |\mathbf{w}\rangle,
$$
where $g^N_m$ is the set of $S^z$-basis states on $N$ sites with magnetization $m$, and $|g_m^N|$ denotes the cardinality of this set.  The ground space of $H$ is spanned by the states $\{|g_m^N\rangle : m = -N, ... , N\}$.  Noting that a $d$-local operator can change the magnetization of $|\mathbf{w}\rangle$ by at most 
$2 d$, define $I = 2 d + 1$ and let $m_{\max}$ be a multiple of $I$ to be chosen later.  We will next show that the subspace
\begin{equation}
\mathcal{C} = \textrm{Span}\left(\left\{|g_m^N\rangle:m=-m_{\max}, -m_{\max} + I, \ldots, m_{\max} - I, m_{\max}  \right\}\right).\label{eq:motzkincodespace}
\end{equation} 
is an AQECC.  For the case $i \neq j$ in Eq.~\eqref{eq:aqec} we note that $\langle g_m^N|E |g_{m'}^N\rangle = 0$ by construction whenever $m$ and $m'$ are distinct codewords, since $|m - m'| > 2d$ and $E$ is a $d$-local error so it can change the magnetization by at most $2d$.  Next for the case $i =j$ in Eq.~\eqref{eq:aqec} we show that an arbitrary observable has approximately the same expectation value for all of the codewords by computing the trace distance between the local reduced density matrices $\rho^{m}_{d} = \tr_{N-d}\left(|g_m^N\rangle\langle g_m^N|\right), \rho^{m'}_{d} = \tr_{N-d}\left(|g_{m'}^N\rangle\langle g_{m'}^N| \right)$ on the support of the error (note that $\tr_{N -d}$ refers to the trace over the $N-d$ sites where the error acts trivially).   As discussed in the introduction we assume that $E$ is an arbitrary error which acts on $d$ consecutive sites.  To compute the reduced density matrices we use the Schmidt decomposition,
\begin{equation}
|g_m^N\rangle = \sum_{r = -d}^{d} \left[\frac{|g^{N-d}_{m-r}||g^{d}_{r}|}{|g^N_m|}\right]^{1/2} |g^{N-d}_{m-r}\rangle |g^d_{r}\rangle, \label{eq:motzkinSchmidt}
\end{equation}
and using Eq.~\eqref{eq:motzkinSchmidt} to compute the reduced density matrices $\rho_d(m,N) , \rho_d(m',N)$, the trace distance between them is
\begin{equation}
\|\rho_{d}(m,N)-\rho_{d}(m',N)\|_1 = \sum_{r = -d}^d|g^{d}_{r}|\left|\frac{|g^{N-d}_{m-r}|}{|g^N_m|}-\frac{|g^{N-d}_{m'-r}|}{|g^N_{m'}|}\right|.\label{eq:motzkinerror}
\end{equation}
The quantity $|g^L_i|$ can be computed as follows:  for $i \geq 0$ we can have a number of flat steps equal to $f = 0,...,L-i$, and for each $f$ these flat steps can be arranged in $\binom{L}{f}$ ways, and to achieve magnetization $i$ we must insert $\frac{L -f +i}{2}$ up steps in $L-f$ positions.  The case $i < 0$ is similar, therefore
\begin{equation}
|{g^L_i}|=\sum_{f=0}^{L-i}\binom{L-f}{(L-f+i)/2}\binom{L}{f}.\label{eq:exactGLI}
\end{equation}
When $i = \mathcal{O}(\sqrt{L})$ an asymptotic expansion for \eq{exactGLI} is given in equation (30) of \cite{movassagh2017entanglement},
\begin{equation}
|g^{L}_{i}|\approx\frac{3^{L+1/2}}{2\sqrt{\pi L}}e^{-3 i^2/4L}.
\end{equation}
Restricting $m_{\max}$ to be $\mathcal{O}(\sqrt{N})$ in \eq{motzkincodespace}, upper bounding the first binomial coefficient with $|g^d_0|$ in Eq.~\eqref{eq:motzkinerror} and applying this asymptotic result to Eq.~\eqref{eq:motzkinerror} yields

\begin{equation}
\|\rho_d(m,N)- \rho_d(m',N)\|_1  \leq  \frac{3^{1/2} \sqrt{N}d}{\sqrt{\pi d (N-d)}} \max^{}_{r}\left|\exp\left[{\frac{3}{4}\left( \frac{ m^2}{N} - \frac{ (m-r)^2}{(N-d)} \right)}\right] - \exp\left[{\frac{3}{4}\left( \frac{ m'^2}{N} - \frac{ (m'-r)^2}{(N-d)} \right)}\right] \right |.
\end{equation}

Taylor expanding the exponentials inside the absolute value,
$$
\left|\exp\left[{\frac{3}{4}\left( \frac{ m^2}{N} - \frac{ (m-r)^2}{(N-d)} \right)}\right] - \exp\left[{\frac{3}{4}\left( \frac{ m'^2}{N} - \frac{ (m'-r)^2}{(N-d)} \right)}\right] \right| \approx  \frac{3}{4} \left |\frac{(m^2 - m'^2)}{N} - \frac{(m-r)^2 -(m' - r)^2}{N-d}\right | .
$$
The term in the absolute value satisfies
\begin{align}
\left |\frac{(m^2 - m'^2)}{N} - \frac{(m-r)^2 -(m' - r)^2}{N-d}\right | &=\left| \frac{2(m-m')r}{N-d} - \frac{r(m^2-m'^2)}{N(N-d)}\right| \\
&\leq \left| \frac{2(m-m')r}{N-d}\right| + \left|\frac{r(m^2-m'^2)}{N(N-d)}\right|,\label{eq:motzkinintermediateexpression}
\end{align}
and since the second term in \eq{motzkinintermediateexpression} is subleading we have
\begin{align}
\|\rho_{d}(m,N)-\rho_{d}(m',N)\|_1 = \mathcal{O}\left(\frac{ \sqrt{N}|m - m'|d^2}{\sqrt{d}(N-d)^{3/2}}\right) = \mathcal{O}\left(\frac{d^{3/2}|m-m'|}{N}\right) = \mathcal{O}\left(\frac{d^{3/2}m_{\max}}{N}\right). \label{eq:scalingOfMotzkinError}
\end{align}

To encode $k = \lfloor \log_2 \textrm{dim}(\mathcal{C}) \rfloor$ logical qubits it suffices to take $m_{\max} = d \cdot 2^k$.  Therefore for any $a, b > 0$ with $5 a/2 + b < 1/2$ the ground space of the spin 1 Motzkin model contains an $[[N, k, d, \varepsilon]]$ AQECC with $k = a \log N$ , $d = b \log N$, and $\varepsilon = \mathcal{O}\left(\frac{\log^{2} N}{N^{1/2 - 5 a/2 - b}} \right)$.
\end{proof}
\end{document}